\documentclass[journal]{IEEEtran}
\ifCLASSINFOpdf
  \usepackage[pdftex]{graphicx}
\else
\fi
\usepackage{amsmath}

\usepackage{amsmath,amssymb,amsfonts,mathtools}

\usepackage{amssymb}
\usepackage[hidelinks]{hyperref}
\usepackage{nameref,cleveref}
\Crefname{thm}{Theorem}{Theorems}
\Crefname{prop}{Proposition}{Propositions}
\Crefname{lem}{Lemma}{Lemmas}
\Crefname{cor}{Corollary}{Corollaries}
\Crefname{defn}{Definition}{Definitions}
\Crefname{assump}{Assumption}{Assumptions}
\Crefname{conj}{Conjecture}{Conjectures}
\Crefname{alg}{Algorithm}{Algorithms}
\Crefname{appsec}{Appendix}{Appendices}
\Crefname{equation}{}{}
\Crefname{figure}{Fig.}{Figs.}
\creflabelformat{figure}{#2\textup{#1}#3}

\usepackage{amsthm}

\theoremstyle{remark}
\newtheorem{thm}{Theorem} 
\newtheorem{defn}{Definition} 

\begin{document}
%
\title{A Closed-form Solution for the Strapdown Inertial Navigation Initial Value Problem}

%
%
%
%
%

    

\author{James Goppert, Li-Yu Lin, Kartik Pant, and Benjamin Perseghetti
\thanks{J. Goppert, L. Lin, and K. Pant are with the School of Aeronautics and Astronautics, Purdue University, West Lafayette, IN 47906, USA (e-mail: jgoppert@purdue.edu; lin1191@purdue.edu;  kpant@purdue.edu).}
\thanks{B. Perseghetti is with Rudis Laboratories, Dayton, OH 45342 USA (e-mail: bperseghetti@rudislabs.com).}}

\maketitle

\begin{abstract}
Strapdown inertial navigation systems (SINS) are ubiquitious in robotics and engineering since they can estimate a rigid body pose using onboard kinematic measurements without knowledge of the dynamics of the vehicle to which they are attached. While recent work has focused on the closed-form evolution of the estimation error for SINS, which is critical for Kalman filtering, the propagation of the kinematics has received less attention. Runge-Kutta integration approaches have been widely used to solve the initial value problem; however, we show that leveraging the special structure of the SINS problem and viewing it as a mixed-invariant vector field on a Lie group, yields a closed form solution. Our closed form solution is exact given fixed gyroscope and accelerometer measurements over a sampling period, and it is utilizes 12 times less floating point operations compared to a single integration step of a 4th order Runge-Kutta integrator. We believe the wide applicability of this work and the efficiency and accuracy gains warrant general adoption of this algorithm for SINS.
\end{abstract}

\IEEEpeerreviewmaketitle

\section{Introduction}

\IEEEPARstart{T}{he} problem of strapdown inertial navigation systems (SINS) has been widely studied in the literature~\cite{titterton2004strapdown, wu2005strapdown, chang2021strapdown}. It involves estimating the position and orientation of a rigid body employing onboard kinematic measurements. These measurements are generally obtained using an inertial measurement unit (comprising an accelerometer and a gyroscope). These measurements provide the acceleration and angular velocity in the body frame. One can use these measurements and perform integration over time to calculate the position and orientation of the rigid body. The existing methods utilize numerical integration approaches such as Runge-Kutta \cite{BOGACKI199615, DORMAND198019} methods to solve this initial value problem \cite{titterton2004strapdown, wu2005strapdown, chang2021strapdown}. However, due to numerical integration errors, the predicted position and orientation can drift over time, even given fixed angular velocity and acceleration in the body frame. To solve this issue, navigation systems utilize additional aides in the form of GNSS sensors, magnetometers, speedometers, etc \cite{titterton2004strapdown}.      

The recent developments in symmetry-preserving observers on matrix Lie groups\cite{bonnabel2008symmetry,van2022equivariant, van2021autonomous} have demonstrated superior performance and stronger stability guarantees for various industrial applications. These applications include attitude estimation on $SO(3)$ \cite{bonnabel2008symmetry, roberts2011attitude, sanyal2008global,panttowards}, pose estimation on $SE(3)$\cite{vasconcelos2010nonlinear,hua2011observer}, homography estimation on $SL(3)$\cite{hamel2011homography}, etc. The recent works \cite{barrau2018invariant, barrau2016invariant, barrau2017three} show the convergence properties of the invariant observers for deterministic systems. The authors employ the log-linear property of the invariant error dynamics to prove the stability and optimality of the filter.
The authors in \cite{chang2021strapdown, chang2023sins} have utilized the special group of double direct isometries $SE_2(3)$ to solve the initial alignment problem for SINS. While recent works on invariant filtering have discovered a closed form for the exact evolution of error from a reference trajectory on $SE_2(3)$\cite{li2022closed, li2023errors, linapplication}, few papers address more efficient and accurate methods for propagating the reference trajectory itself.
In this work, we focus on the propagation step of the kinematics of the system with underlying symmetry properties. We provide a closed-form expression for the propagation model for a class of systems whose dynamics evolve on a matrix Lie group: mixed invariant systems. Mixed-invariant systems \cite{khosravian2016state, lin2022log-linear} can be used to describe a wide variety of kinematics systems. One of the most useful applications is in the description of rigid body motion. Here, the right-invariant vector field describes the acceleration of gravity, and the left-invariant vector field describes the body frame acceleration and angular velocity. In previous approaches, the evolution of SINS is being evaluated using numerical integration. We show that the theory of mixed-invariant systems on Lie groups can be leveraged to derive a closed-form solution for the SINS propagation problem. Since our method is exact, the time step has no impact on the accuracy of the solution. In addition, the computational complexity is 12 times less than for solving the same initial value problem with a Runge-Kutta 4th-order solver.

The rest of the paper is organized as follows. In \Cref{sec:math} we provide the fundamental elements necessary for the description of our solution using mixed-invariant vector fields on Lie groups. In \Cref{sec:problem}, we formulate the SINS initial value problem as a mixed-invariant vector field. In \Cref{sec:contrib} we provide our closed-form solution for the SINS initial value problem. In \Cref{sec:sim} we show the comparison with Runge-Kutta integration schemes. Finally, \Cref{sec:conclu} summarizes the paper and proposes future directions.

\section{Mathematical Preliminaries}
\label{sec:math}
This section presents a concise description of the $SE_2(3)$ matrix Lie group and its associated Lie algebra used for deriving the closed-form propagation of the mixed-invariant dynamics of the SINS problem. A detailed description of matrix Lie groups can be found in \cite{hall2013lie,lee2012smooth}.

The group of double direct spatial symmetries $SE_2(3)$, can be represented as:
\begin{equation}
    X \coloneqq \begin{bmatrix}
        R & v & p \\
        0 & 1 & 0 \\
        0 & 0 & 1
    \end{bmatrix} \coloneqq 
    \begin{bmatrix}
        R & P \\
        0_{2\times 1} & I_2
    \end{bmatrix}  
\end{equation}
where $R \in SO(3)$ denotes the attitude of the body with respect to the world frame, $p \in \mathbb{R}^{3}$ denotes the position in the world frame, and $v \in \mathbb{R}^{3}$ denotes the translational velocity in the world frame. We write this as a block matrix for ease of manipulation later, where $P \equiv \begin{bmatrix}v\ p \end{bmatrix}$.

The corresponding $\mathfrak{se}_2(3)$ Lie algebra can be represented as:
\begin{equation}
    [x]^\wedge \coloneqq \begin{bmatrix}
        [\omega]^\wedge & u_1 & u_2 \\
        0 & 0 & 0 \\
        0 & 0 & 0 
    \end{bmatrix} \coloneqq
    \begin{bmatrix}
        \Omega & U \\
        0_{2\times1} & 0_{2\times 2}
    \end{bmatrix}   
\end{equation}
where $x = \begin{bmatrix}u_1 & u_2 & \omega\end{bmatrix}^T$, and $[\cdot]^\wedge$ indicates the wedge operator that maps elements from $\mathbb{R}^n$ to a corresponding Lie algebra with dimension $n$. $\omega \in \mathbb{R}^{3}$ denotes the angular velocity in the body frame, $\Omega \in \mathfrak{so}(3)$ is the corresponding skew-symmetric matrix of $\omega$, and $u_1 \in \mathbb{R}^{3}$, $u_2 \in \mathbb{R}^{3}$. We write this as a block matrix for ease of manipulation later, where $U \equiv \begin{bmatrix}u_1\ u_2 \end{bmatrix}$, $\Omega \equiv [\omega]^\wedge$

\section{Problem Formulation}
\label{sec:problem}
SINS propagation requires solving the initial value problem with kinematics given by: 
\begin{equation}
\label{eq:SINS}
\begin{aligned}
    \dot{R} = R \omega &&
    \dot{v} = R a + a_g &&
    \dot{p} = v
\end{aligned}
\end{equation}
where $a = \begin{bmatrix} a_x & a_y & a_z\end{bmatrix}^T$ represents the translational acceleration in the body frame, and $a_g = \begin{bmatrix} 0 & 0 & g\end{bmatrix}^T \in \mathbb{R}^{3}$ represents the fixed gravitational acceleration in the world frame. Note that our approach can be generalized to any $a$, and $a_g$.

These kinematics an be described using a mixed-invariant vector field.

\begin{defn}
A mixed-invariant system~\cite{khosravian2016state, lin2022log-linear} evolving on a Lie group $G$ is a system with the differential equation:
\begin{align}
    \dot{X} = M X + X N
\end{align}
where $X \in G$, and $M$, and $N$ can be any constant matrices, satisfying: $Ad_{X^{-1}} M + N \in \mathfrak{g}$, where $\mathfrak{g}$ is the associated Lie algebra of $G$.
\end{defn}

The solution of $\dot{X} = M X + X N$ is $X_t = e^{Mt} X(0) e^{Nt}$. $X(0) = \begin{bmatrix} R_0 & P_0 \\ 0 & I\end{bmatrix}$ is the initial value of the states, where $R_0$ is the initial rotation matrix and $P_0$ represents the initial position and translational velocity in the world frame. The challenge in finding a closed-form solution for the initial value problem is finding the matrix exponential of $e^{MT}$ and $e^{Nt}$.

For SINS, we will leverage the special structure of the M and N matrices.

\begin{align}
M \coloneqq \begin{bmatrix}
0_{3\times 3} & A_M\\
0_{2\times 3}& -B \\
\end{bmatrix}
\end{align}

\begin{align}
N \coloneqq \begin{bmatrix}
\Omega & A_N\\
0_{2\times 3} & B
\end{bmatrix}
\end{align}

\begin{align}
A_M \coloneqq \begin{bmatrix}
0 & 0 \\
0 & 0 \\
g & 0
\end{bmatrix}
A_N \coloneqq \begin{bmatrix}
a_x & 0 \\
a_y & 0 \\
a_z & 0
\end{bmatrix}
B \coloneqq \begin{bmatrix}
0 & 1 \\
0 & 0
\end{bmatrix}
\end{align}

\section{Main Contribution}
\label{sec:contrib}
Note, that achieving a closed-form solution is tractable since B is a nil-potent matrix. We will leverage this in our derivation of \Cref{thm:thm1}.

\begin{thm}
\label{thm:thm1}
The solution of the SINS initial value problem defined in \cref{eq:SINS} is given by:
\begin{align}
X(t) = \begin{bmatrix}
R_{r'} R_0 R_{l'} & R_{r'} R_0 P_N + (R_{r'} R_0 + P_M)(I + Bt) \\
0 & I
\end{bmatrix}
\end{align}
where:
\begin{align*}
&R_{l'} \coloneqq e^{\Omega t} \hspace{0.2cm} R_{r'} \coloneqq I_3 \\
&P_M \coloneqq P(0, A_Mt, -Bt) \hspace{0.2cm} P_N \coloneqq P(\Omega t, A_Nt, Bt) \\
&P(\Omega, A, B) = A + AB/2 + \Omega(C_1 I_2 + C_2 B) \\ &\hspace{2 cm} + \Omega^2 A (C_2 I_2 + C_3 B) \\
&\theta \coloneqq \sqrt{\omega^T \omega} \hspace{0.2cm} C_1 \coloneqq \dfrac{1 - \theta^2/2 - \cos(\theta)}{\theta^2}\\
&C_2  \coloneqq \dfrac{\theta - \sin{\theta}}{\theta^3} \hspace{0.2cm} C_3  \coloneqq \dfrac{\theta^2/2 - \theta^4/24 + \cos{\theta} - 1}{\theta^4}
\end{align*}

\end{thm}

\begin{proof}
We wish to find a closed form for $X(t) = e^{lt} X(0) e^{rt}$. We note that $l$ and $r$ are not elements of the Lie algebra; however, they are of the form below:

\begin{equation*}
l \coloneqq \begin{bmatrix}
    \Omega && A\\
    0_{2\times3} && B
\end{bmatrix}
\end{equation*}
Since B is nil-potent, higher powers are the zero matrix, the square of $l$ is:
\begin{equation*}
l^2 = \begin{bmatrix}
    \Omega^2 && \Omega A + AB\\
    0_{2\times3} && B^2
\end{bmatrix}
\end{equation*}
for $n > 1$, the power of $l$ can be written as:
\begin{equation*}
l^n = \begin{bmatrix}
    \Omega^n && \Omega^{n-1} A + \Omega^{n-2}AB \\
    0_{2\times3} && B^n
\end{bmatrix}
\end{equation*}
We can write the exponential of $e^{l}$ as the series:
\begin{equation*}
e^{l} = \sum\limits_{n=0}^{\infty} \dfrac{l^n}{n!} = \begin{bmatrix}
    \sum\limits_{n=0}^{\infty} \dfrac{\Omega^n}{n!} & N \\
    0_{2\times3} & \sum\limits_{n=0}^{\infty} \dfrac{B^n}{n!}
\end{bmatrix}
\end{equation*}
where $N = A + \sum\limits_{n=2}^{\infty} \left[\dfrac{\Omega^{n-1} A}{n!} + \dfrac{\Omega^{n-2}AB}{n!}\right]$
Which can be rewritten as a summation from $0$ to $\infty$:
\begin{equation*}
N = A + \sum\limits_{n=0}^{\infty} \left[\dfrac{\Omega^{n + 1} A}{(n+2)!} + \dfrac{\Omega^{n}AB}{(n+2)!}\right] 
\end{equation*}
Now, we split the summation into even and odd terms:
\begin{align*}
&N = A + \sum\limits_{n=0}^{\infty}\left[ \dfrac{\Omega^{2n + 1} A}{(2n+2)!} + \dfrac{\Omega^{2n}AB}{(2n+2)!}\right] \\
& \hspace{1 cm} + \sum\limits_{n=0}^{\infty}\left[\dfrac{\Omega^{2n + 2} A}{(2n+3)!} + \dfrac{\Omega^{2n + 1}AB}{(2n+3)!} \right]     
\end{align*}
Using the property of the skew-symmetric matrices for $\Omega$ :
\begin{align*}
&\sum\limits_{n=0}^{\infty}  \dfrac{\Omega^{2n + 1} A}{(2n+2)!} =  \sum\limits_{n=0}^{\infty} \dfrac{(-1)^{n} \theta^{2n} \Omega A}{(2n+2)!} = C_1 \Omega A \\
&\sum\limits_{n=0}^{\infty}  \dfrac{\Omega^{2n+2} A}{(2n+3)!} = \sum\limits_{n=0}^{\infty} \dfrac{(-1)^{n} \theta^{2n} \Omega^2 A}{(2n+3)!} = C_2 \Omega^2 A\\
&\sum\limits_{n=0}^{\infty}  \dfrac{\Omega^{2n+1} AB}{(2n+3)!} = \sum\limits_{n=0}^{\infty} \dfrac{(-1)^{n} \theta^{2n} \Omega AB}{(2n+3)!} = C_2 \Omega AB\\
&\sum\limits_{n=0}^{\infty}  \dfrac{\Omega^{2n} AB}{(2n+2)!} =  \dfrac{AB}{2} + \sum\limits_{n=1}^{\infty} \dfrac{\Omega^{2n} AB}{(2n+2)!} \\ &\hspace{1 cm} =  \dfrac{AB}{2} + \sum\limits_{n=0}^{\infty} \dfrac{\Omega^{2n+2} AB}{(2n+4)!} \\
&\hspace{1 cm} =  \dfrac{AB}{2} + \sum\limits_{n=0}^{\infty} \dfrac{(-1)^{n} \theta^{2n} \Omega^{2} AB}{(2n+4)!} = \dfrac{AB}{2} + C_3 \Omega^2 A B 
\end{align*}
We now arrive at the closed form of $P(\Omega, A, B)$:
\begin{align*}
P(\Omega, A, B) &= A + AB/2 + \Omega A (C_1 I_2 + C_2B) \\
&+ \Omega^2 A (C_2 I_2 + C_3 B)\end{align*}
The complete matrix exponential of $e^{lt}$ and $e^{rt}$ may now be computed by letting $l'=lt$ and $r'=rt$:
\begin{align*}
e^{lt} &= \begin{bmatrix}
    R_{l'} & P_N \\
    0 & I + Bt
\end{bmatrix} 
e^{rt} = \begin{bmatrix}
    R_{r'} & P_M \\
    0 & I - Bt
\end{bmatrix}
\end{align*}
%
and the solution of $X(t)$ for strapdown inertial navigation propagation can be found:
\begin{align*}
X(t) &= e^{rt}X(0)e^{lt} \\ &= \begin{bmatrix}
R_{r'} R_0 R_{l'} & R_{r'} R_0 P_N + (R_{r'} R_0 + P_M)(I + Bt) \\
0_{2\times3} & I_2
\end{bmatrix}
\end{align*} 
\end{proof}

\section{Simulation Study}
\label{sec:sim}
In order to compare our closed-form solution in \Cref{thm:thm1} to a Runge-Kutta 4th order algorithm, we leveraged Casadi\cite{Andersson2019} to build equation graphs of both the Runge-Kutta 4th order integrator and our Lie group mixed-invariant closed-form solution. Our source code is open source and available for download \href{https://github.com/CogniPilot/cyecca}{https://github.com/CogniPilot/cyecca}.

For testing, we compared against a known solution of a circular trajectory in the presence of gravity. The velocity of the vehicle was $1$ m/s, the radius of the circular trajectory was $1$ m. The trajectory of the vehicles can be seen in \Cref{fig:trajectory}.

There was a significant 12x reduction in the number of floating point operations per propagation step. As seen in \Cref{fig:ops}. These operations were counted using the Casadi equation graph itself.

\begin{figure}
    \centering
    \includegraphics[width=\linewidth]{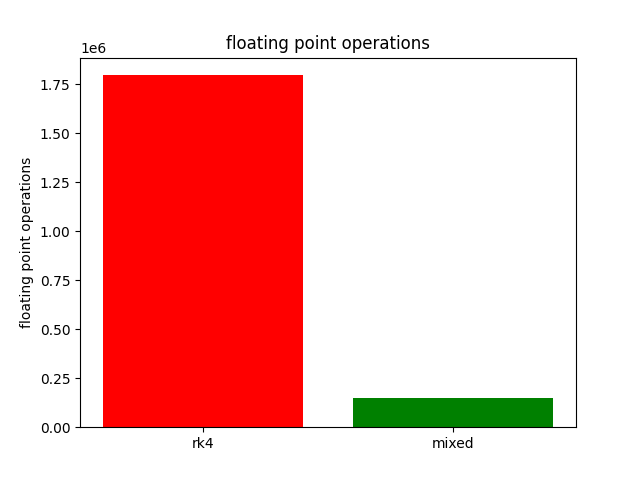}
    \caption{12 times fewer floating point operations are required for our closed form "mixed" invariant method vs. rk4 method as evaluated by counting operations in the Casadi equation graph}
    \label{fig:ops}
\end{figure}

\begin{figure}
    \centering
    \includegraphics[width=\linewidth]{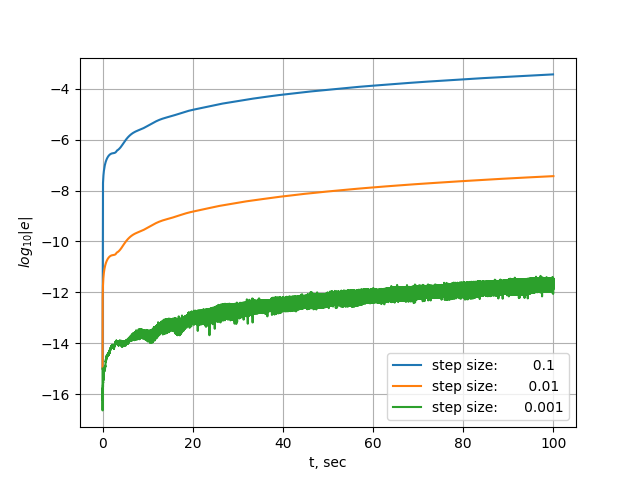}
    \caption{As the step size of the rk4 method is decreased, it converges to the ground truth and our closed form solution.}
    \label{fig:rk-step-err}
\end{figure}

\begin{figure}
    \centering
    \includegraphics[width=\linewidth]{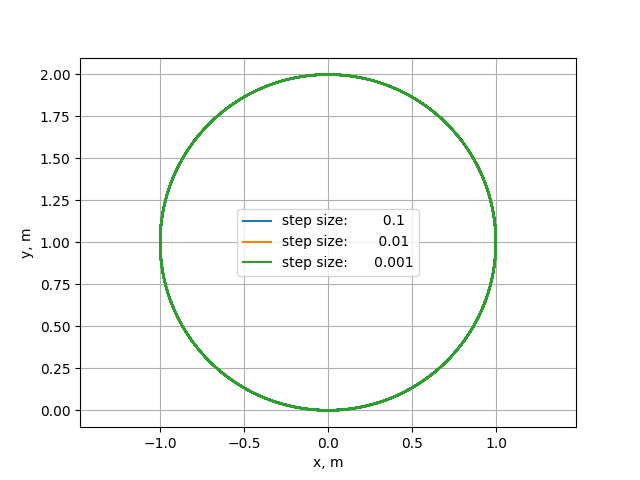}
    \caption{This is a circular ground truth trajectory used to compare the trajectories. At this scale, there are no visible discrepancies in the rk4 step size.}
    \label{fig:trajectory}
\end{figure}

\section{Conclusion}
\label{sec:conclu}

In this work, we derived an efficient and exact solution of the SINS initial value problem leveraging mixed-invariant vector fields on Lie groups. Given a fixed angular velocity and acceleration in the body frame, and a fixed acceleration in the world frame (gravity), we can exactly predict the evolution of the system. Compared to the industry standard Runge-Kutta 4th order integration scheme, our approach is more accurate and is also more efficient with 12 times fewer floating point operations necessary.
Given our block matrix derivation, we also note that this closed form can be generalized to the $SE_n(3)$ Lie group given mixed-invariant kinematics of a form where $B$ is a nilpotent matrix which could be of interest for future research.


%


\section*{Acknowledgment}

The authors would like to thank NXP for their support and contribution to the CogniPilot Foundation to help enable this work.

\bibliographystyle{IEEEtran}
\bibliography{ref.bib}









\end{document}